\documentclass[a4paper,UKenglish,cleveref, autoref, thm-restate]{lipics-v2021}

\usepackage{amssymb}
\usepackage{mathtools}
\usepackage{graphicx}
\usepackage[ruled,linesnumbered,vlined]{algorithm2e}
\usepackage{microtype}
\usepackage{amsmath}
\usepackage{amsthm}
\usepackage{enumerate}
\usepackage{wrapfig}
\usepackage{multirow}
\usepackage{nicefrac}
\usepackage[colorinlistoftodos,prependcaption,textsize=small]{todonotes}
\usepackage{multicol}
\usepackage{thm-restate}
\usepackage{microtype}
\usepackage{url}
\usepackage{wrapfig}
\usepackage{multirow}
\usepackage{algorithm2e}
\usepackage{algorithm2e}
\usepackage{makecell}
\usepackage{float}

\definecolor{forestgreen}{rgb}{0.13, 0.55, 0.13}

\SetCommentSty{mycommfont}

\newtheorem{problem}{Problem}

\graphicspath{{./figures/}}

\def\reals{{\mathbb R}}

\def\eps{{\varepsilon}}

\newcommand{\frechet}{Fr\'echet}
\newcommand{\df}{d_{F}}
\newcommand{\dfd}{\text{d}_{dF}}

\def\NNG{\text{NNG}}
\DeclareMathOperator{\MST}{MST}
\DeclareMathOperator{\TSP}{TSP}

\def\b{\textbf{b}}
\def\Geps{G_{\eps}}
\def\Gepso{G_{\eps^*}}

\def\dfn#1{\emph{\textcolor{forestgreen}{\textbf{#1}}}}

\date{}

\newcommand{\old}[1]{{{}}}
\newcommand{\temp}[1]{{{}}}
\newcommand{\fullversion}[1]{{{}}}

\newcommand{\oldparagraph}[1]{\vspace{5pt}\noindent\textbf{#1}}

\title{On \frechet\ Traveling Salesmen Problems}

\author{Omrit Filtser} {Department of Mathematics and Computer Science, The Open University of Israel, Ra'anana, Israel} {omrit.filtser@gmail.com} {https://orcid.org/0000-0002-3978-1428}{}

\author{Tzalik Maimon} {Department of Mathematics and Computer Science, The Open University of Israel, Ra'anana, Israel} {post.tmx@gmail.com} {https://orcid.org/0009-0000-8459-8553}{}

\author{Michal Moiseev} {Department of Mathematics and Computer Science, The Open University of Israel, Ra'anana, Israel} {mikkamois@gmail.com} {}{}

\authorrunning{O. Filtser, T. Maimon and M. Moiseev} 

\Copyright{Omrit Filtser, Tzalik Maimon and Michal Moiseev} 



\relatedversion{}

\acknowledgements{We would like to thank an anonymous reviewer for helpful comments and references, and in particular for suggesting the expected linear time algorithm in Section 3.}

\nolinenumbers

\begin{document}

\maketitle

\begin{abstract}
    The \frechet\ distance is a well-studied distance measure between two curves.
    In this work, we demonstrate that the merit of \frechet\ distance extends beyond evaluating similarity, and introduce a new setting in which it proves useful.
    Consider a situation where two agents are required to visit a given set of sites, while staying close to each other throughout their traversal. In this paper, we study problems where the goal is to construct two curves whose vertices are from a given set of points, under the constraint that the \frechet\ distance between the curves is kept as small as possible. This problem can be viewed as a variant of the Traveling Salesman Problem (TSP), and thus may be of interest in routing, network planning and more.
    We present a near-linear algorithm for this problem under the discrete \frechet\ distance, and explore several variants of the problem, including minimizing the lengths of the curves and balancing the number of sites assigned to each agent. Lastly, we prove that the problem is NP-hard under the continuous \frechet\ Distance. 
\end{abstract}

\newpage

\section{Introduction}
The traveling salesman problem (TSP) asks for the shortest route that visits a given set of points. It is a classic NP-hard problem, and have various approximation algorithms (including a PTAS in the Euclidean plane~\cite{Arora98,Mitchell99}). In some generalizations of TSP, e.g. the Vehicle Routing Problem or Multiple-TSP (see, e.g. \cite{BEK2006,BBJWG20,CK2021,MDMM21}), the goal is to plan short routes for a group of agents to visit all the points together. In these variants there is no restriction on the relations between the routes. 

Consider a situation where two agents are asked to visit a given set of sites, they can split but they still need to remain close to each other throughout their motion. This can be required for example in order to keep then in some range of communication, or so that one agent can quickly get to other in case of trouble. Therefore, in this paper, we introduce and study a new set of problems where the goal is to plan routes for two (or more) agents, that together visit a given set of sites, while keeping the agents close to each other throughout the traversal. More precisely, our goal is to construct two curves from a given set of points, such that the (continuous or discrete) \frechet\ distance between the curves is as small as possible. We also consider variations of this \frechet-TSP type problem where the goal is to balance the load on the agents, in different ways.

The \frechet\ distance~\cite{F1906}, introduced by \textit{Maurice \frechet}, is a popular measure of similarity between curves. It is often described by an analogy of a person and a dog connected by a leash, both walking forward along two separate curves, while the leash keeps them at a bounded distance. The \frechet\ distance is then the shortest length of a leash that allows them to traverse their respective curves. The \frechet\ distance takes into account both the position and ordering of points, distinguishing it from other metrics like the Hausdorff distance. This property makes the \frechet\ distance useful across a wide range of applications.

The discrete \frechet\ distance focuses solely on the points of the curves rather than the edges between them. This variant can be described analogously by replacing the man and dog by two frogs that are ``hopping'' along the vertices of the curves while attempting to maintain a small distance between them. This discrete approach allows simpler and (slightly) faster algorithms for computing the distance.

Alt and Godau~\cite{AHGM1995} showed that the \frechet\ distance between a curve $P$ of length $n$ and a curve $Q$ of length $m$, can be calculated in $O(nm\log(nm))$ time. Eiter and Mannila~\cite{EM94} showed that the discrete \frechet\ distance can be computed in $O(mn)$ time. Only 20 years later, it has been shown that under the Strong Exponential Time Hypothesis (SETH), the \frechet\ distance cannot be computed in strongly subquadratic time \cite{B2014}. Very recently, Cheng, Huang, and Zhang~\cite{CHZ25}, presented a strong subquadratic time algorithm that computes a $(7+\eps)$-approximation for both the discrete and continuous distance.

A plethora of applications and variants of the \frechet\ distance have been studied since then (see, e.g.~\cite{BBMS19,DH13,EFV07,MSSZ11}).
The variant most relevant to our work is the \frechet\ distance between two point sets, recently studied by Buchin and Kilgus~\cite{BK22}. 
In this variant, the input objects are two sets of points rather than two curves, and the goal is to construct two curves (one for each set of points), to minimize the \frechet\ distance between them. The main difference from our version is that we also need to find a partition of the point set into two sets. 
Buchin and Kilgus showed that under the discrete \frechet\ distance, the problem is equivalent to computing the Hausdorff distance between the two sets, and thus can be solved in $O(nm\log(nm))$ time (where $n,m$ are the sizes of the two sets).
The continuous version, on the other hand, was shown to be NP-complete.
They also provide an exponential time algorithm running in $O(k^a((m+n-a)+a\log a))$ where $a$ denotes the number of points that can be matched only to edges (``floating'' points in their terminology) and $k$ is the maximum number of edges that can cover such a point.

Another closely related variant is the Curve-Point-Set Matching Problem (CPSM), where given a polygonal curve $P$, a set of points $S$, and a maximum distance $\delta > 0$, the objective is to find another polygonal curve $Q$, whose set of vertices is either a subset of $S$ or contains all the points of $S$ (also, either with or without repetitions, referred to as the non-unique/unique variant, respectively), such that the continuous or discrete \frechet\ distance between the new curve $Q$ and the original curve $P$ is smaller than $\delta$. Maheshwari, Arora, and Smid~\cite{MAS11} addressed the continuous non-unique subset version of the problem, and presented an algorithm that runs in $O(nk^2)$ time. Wylie~\cite{WT13} studied CPSM under the discrete \frechet\ distance, and proved that the unique version (where each point in $S$ can appear only once in $Q$) is NP-complete for both the subset and all-point versions. This is in contrast to the non-unique cases which were shown to be polynomial-time solvable (in $O(nk)$ time, where $n$ is the size of $S$ and $k$ the size of $P$). Accisano and Ungor \cite{AU12} also showed NP-completeness under the continuous distance for the all-points variant, both in unique and non-unique settings.

\oldparagraph{Overview.} 
To the best of our knowledge, the \frechet-TSP problem described above has not been studied before. This gives rise to a rich family of TSP-style questions centered around the \frechet\ distance, and we systematically outline several natural variants that capture different aspects of this setting.
In \Cref{sec:discrete-frechet-tsp}, we present an algorithm for the discrete \frechet-TSP problem that runs in $O(n\log n)$ time for a set of $n$ points in a constant dimension. We utilize a combination of the properties of Nearest-Neighbor-Graph and Minimum-Spanning Tree (MST) over a Unit-Disk Graph (UDG).
By utilizing the properties of UDG, we ensure an upper bound on the distance between the partitions we output. By utilizing the properties of the MST, we obtain a linear time partitioning. The initialization of these constructs are obtained in $O(n \log n)$ time.
In \Cref{sec:min-length}, we consider a variant of the problem aiming to minimize the total length of the curves, and present a constant-factor approximation algorithm that also runs in $O(n\log n)$ time.
In \Cref{sec:balance}, we focus on balancing the number of vertices among the partitioned components, and present different strategies to achieve almost optimal balance. Here, we utilize the properties of the UDG again. Specifically, the upper bound on the number of kissing number of the graph.
This in turn limits the number of cases in the possible partition we wish to balance. We handle these case-by-case showing that we can get the balance as close as at most 1 from optimal.
Finally, in \Cref{sec:continuous}, we show that the continuous variant of \frechet-TSP is NP-hard. Out proof is an adaptation of the method used by Buchin and Kilgus \cite{BK22}.

\section{Notations and problem definition}\label{sec:prelims}

A polygonal curve $P$ in $\reals^d$ is a continuous function $P:[1,n]\rightarrow\reals^d$, such that for any integer $1\le i\le n-1$ the restriction of $P$ to the interval $[i,i+1]$ is a straight line segment. The points $P[0],\dots,P[n]$ are the vertices of $P$, and the segments $\overline{P[i]P[i+1]}$ are the edges of $P$. 

Let $P:[0,n]\rightarrow\reals^d$ and $Q:[0,m]\rightarrow\reals^d$ be two polygonal curves with number of vertices $n$ and $m$, respectively.

\oldparagraph{The (continuous) \frechet\ distance.} A reparameterization of a curve $P$ is a continuous, non-decreasing, surjective function $f:[0,1]\rightarrow[1,n]$ such that $f(0)=1$ and $f(1)=n$. The \frechet\ distance between $P$ and $Q$ is defined as
$\df(P,Q) =  \inf_{f,g}{\max_{t\in[0,1]}\|P(f(t))-Q(g(t)}\|,$
where $f$ is a reparameterization of $P$ and $g$ is a reparameterization of $Q$.

\oldparagraph{The discrete \frechet\ distance.}\footnote{For the simplicity of presentation, we follow the definition given in \cite{BJWYZ08}, which is equivalent to the original definition given in \cite{EM94}.}
A \dfn{paired-walk} along $P$ and $Q$ is a sequence of pairs $\pi=\{(P_{i},Q_{i})\}_{i=1}^{k}$, such that $P_{1},...,P_{k}$ and $Q_{1},...,Q_{k}$ partition $P$ and $Q$, respectively, into (disjoint) non-empty subsequences of their vertices, and for any $i$ it holds that either $|P_{i}|=1$ or $|Q_{i}|=1$.
The \dfn{cost} of a paired walk $W$ along $P$ and $Q$ is 
$$\text{cost}(\pi)=\underset{i}{\max}\underset{(p,q)\in P_{i}\times Q_{i}}{\max}\|p - q\|.$$
For any pair $(p,q)\in P_{i}\times Q_{i}$, we say that $p$ is \dfn{matched} to $q$ in $\pi$.

The \dfn{discrete \frechet\ distance} between $P$ and $Q$ is $\dfd(P,Q)=\underset{\pi\in\Pi}{\min}\ \text{cost}(\pi)$,
where $\Pi$ is the set of all possible paired-walks along $P$ and $Q$.

\oldparagraph{Partitioning a point set into curves.}
Let $\delta$ be a \frechet-based distance measure for two curves in $\reals^d$ (i.e., either the continuous or discrete \frechet\ distance). Given a set $S$ of points in $\reals^d$, we say that two curves $P,Q$ \dfn{partition} $S$ if there is a partition of $S$ into two sets $A$ and $B$ such that the set of vertices of $P$ is exactly $A$ and the set of vertices of $Q$ is exactly $B$, and each point in $S$ is used exactly once (either in $P$ or in $Q$). 
The basic version of the problem that we wish to consider is the following.
\begin{problem}[\frechet-TSP]\label{prb:F-TSP}
    Given a set $S$ of points in $\reals^d$, find two curves $P,Q$ that partition $S$ and such that $\delta(P,Q)$ is minimized.
\end{problem}

Note that if we do not require each point in $S$ to be visited by exactly one of the agents, then the problem becomes trivial: we can pick any order on the points in $S$ and set $P=Q$, so the distance between the paths is $0$. 
In \Cref{sec:multiset}, we discuss a variant in which we allow a point to be used more than once by the same agent, and show that this variant is equivalent to \Cref{prb:F-TSP} under the discrete \frechet\ distance.

Denote by $\eps^*$ the distance between the curves in an optimal solution for the \frechet-TSP problem. Clearly, there may be many different optimal solutions, all having distance $\eps^*$. However, some solutions may be considered better than others, for example, if each curve covers roughly the same number of points from $S$ (balanced partition - see \Cref{fig:optimal-frechet-unbalanced}), or if the min-max length (or sum of lengths) of the curves is very small in comparison to other solutions (see \Cref{fig:optimal-frechet-with-long-curves}). We therefore define below different variants of the problem in which the goal is to find a ``good'' solution among those that achieve the optimal distance $\eps^*$.

\begin{figure}[H]
    \centering
    \includegraphics[scale=0.4]{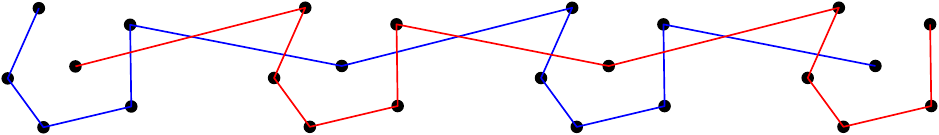}
    \hspace{1cm}
    \includegraphics[scale=0.4]{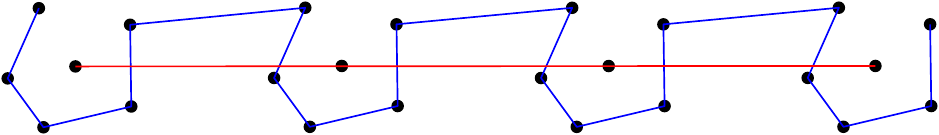}
    \caption{Left: a balanced partition, right: an imbalanced partition. Both have the same distance.}
    \label{fig:optimal-frechet-unbalanced}
\end{figure}

\begin{figure}[H]
    \centering
    \includegraphics[scale=0.3]{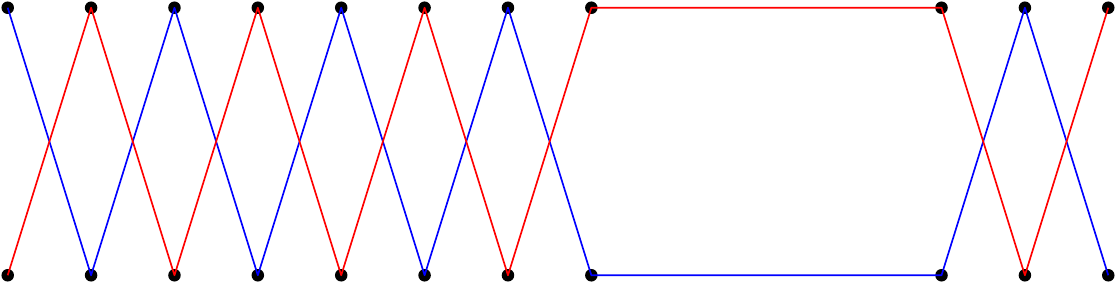}
    \hspace{1cm}
    \includegraphics[scale=0.3]{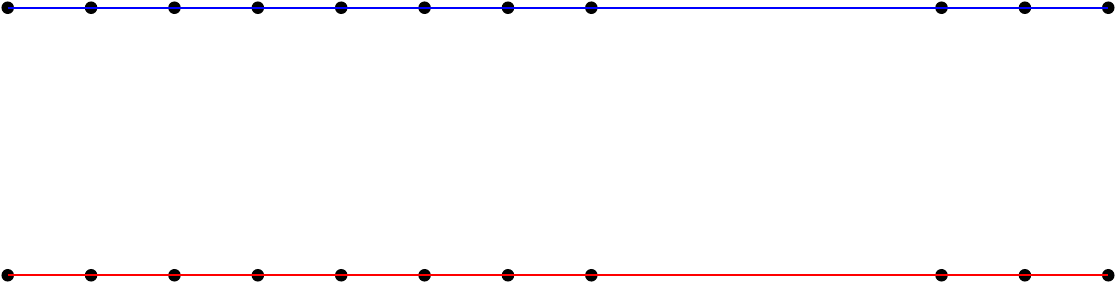}
    \caption{Left: two long curves, right: two shorter curves. Both have the same distance.}
    \label{fig:optimal-frechet-with-long-curves}
\end{figure}

\begin{problem}[balanced-\frechet-TSP]\label{prb:balanced-F-TSP}
    Given a set $S$ of $n$ points in $\reals^d$, find two curves $P,Q$ that partition $S$, such that $\delta(P,Q)=\eps^*$, and $\max\{|P|,|Q|\}$ is minimized.
\end{problem}

For a curve $P$, denote by $\ell(P)$ the sum of the lengths of the edges of $P$. 
\begin{problem}[min-max-\frechet-TSP]\label{prb:min-max-F-TSP}
    Given a set $S$ of points in $\reals^d$, find two curves $P,Q$ that partition $S$, such that $\delta(P,Q)=\eps^*$, and $\max\{\ell(P),\ell(Q)\}$ is minimized.
\end{problem}
\begin{problem}[min-sum-\frechet-TSP]\label{prb:min-sum-F-TSP}
    Given a set $S$ of points in $\reals^d$, find two curves $P,Q$ that partition $S$, such that $\delta(P,Q)=\eps^*$, and $\ell(P)+\ell(Q)$ is minimized.
\end{problem}
Minimizing the length of the path is NP-hard, similar to the traveling salesman problem (TSP), which is NP-hard: for a reduction, simply double each point in an instance of TSP, leading to the following theorem: 

\begin{corollary}
    min-max-\frechet-TSP and min-sum-\frechet-TSP are NP-hard.
\end{corollary}
We can thus aim to find an approximation algorithm for minimizing the length.

\section{Discrete \frechet-TSP}\label{sec:discrete-frechet-tsp}
In this section, we focus on the discrete \frechet\ distance (i.e. $\delta = \dfd$).
We begin by presenting an algorithm for the decision version of the problem: Given a set $S$ of $n$ points in $\reals^d$, and a threshold $\eps \ge 0$, decide whether there exist two curves $P$ and $Q$ that partition $S$ and have $\dfd(P,Q) \leq \eps$.

Let $\Geps = (S,E)$ be the graph whose vertices are the points in $S$, and there is an edge $\{v,u\} \in E$ if and only if $\|v - u\| \leq \eps$. By definition, $\Geps$ is a unit-disk graph with radius $\eps$.
Let $P,Q$ be two curves that partition $S$ s.t. the edges in $P \cup Q$ are edges from $\Geps$. Any paired walk $\pi$ along $P$ and $Q$ with $\text{cost}(\pi)\le\eps$ can be reduced to a set of disjoint stars in $\Geps$ simply by removing edges from $\pi$.
Therefore, if there is such a walk $\pi$, then there is such a set of disjoint starts and vice versa. Thus, we can search for such a set.

\begin{lemma} \label{lem:g_eps_no_degree_0_exists_division}
    There exist two curves $P,Q$ that partition $S$ and have $\dfd(P,Q) \le \eps$ if and only if $\Geps$ does not contain a vertex of degree $0$.
\end{lemma}
\begin{proof}
    If $\Geps$ contains a vertex $v$ of degree $0$, then there is no $u\in S$ with $\|v-u\| \le \eps$. Therefore, $v$ cannot be matched to any other point in $S$ in a paired-walk of cost at most $\eps$.    
    For the other direction, let $C$ be a connected component of $\Geps$. We now show that if $C$ contains more than one vertex, then we can construct two curves on the set of vertices of $C$ as required. This finishes the proof, as we can concatenate the curves that were constructed for all the connected components, and get two curves with vertices from $S$ such that each point in $S$ is used exactly once.    
    Let $T_C$ be some spanning tree of $C$. We color the vertices of $T_C$ red and blue, as follows: first, color all the leaves in red, and then, color the parents of all those leaves in blue. This coloring defines a set of stars, each has a blue center node and at least one red node. By removing these stars from $T_C=T_0$, we are left with a smaller tree, $T_1$. If $T_1$ is a single vertex, then color it red and add it to one of the stars of its child nodes. Otherwise, $T_1$ contains at least one edge. We then repeat the process on $T_1$ and obtain another set of stars, remove them from $T_1$ and get a smaller tree $T_2$. We continue this process until all the vertices are colored. This process results in a partition of the nodes of $T_C$ into stars $S_1,\dots,S_k$, each star contains exactly one blue node $\textbf{blue}(S_i)$, and a non empty sequence of red nodes $\textbf{red}(S_i)$. Let $P=\{\textbf{blue}(S_1),\dots,\textbf{blue}(S_k)\}$ and $Q=\{\textbf{red}(S_1),\dots,\textbf{red}(S_k)\}$. Since for each $i\in[k]$ we have $\|\textbf{blue}(S_i)-v\|\le \eps$ for every $v\in\textbf{red}(S_i)$, the sequence of stars corresponds to a paired-walk with cost at most $\eps$, and therefore we get $\dfd(P,Q)\le \eps$.
\end{proof}

Note that \Cref{lem:g_eps_no_degree_0_exists_division} provides an algorithm for computing a partition of $S$ into two curves $P,Q$ such that $\dfd(P,Q) \le \eps$. Given a spanning forest of $\Geps$, the running time for constructing $P, Q$ is $O(n)$, since each tree can be colored using a a simple BFS traversal. We conclude this in the following corollary.

\begin{corollary}\label{cor:decision-alg}
    Given a spanning forest $T$ of $\Geps$ such that no vertex in $T$ has degree $0$, a partition $P,Q$ of $S$ with $\dfd(P,Q)\le \eps$ can be constructed in $O(n)$ time.
\end{corollary}

We wish to find the partition $P,Q$ of $S$ that minimizes $\dfd(P,Q)$. Denote by $\eps^*$ the distance between the curves in an optimal partition, i.e., there exists a partition $P^*, Q^*$ of $S$ with $\dfd(P^*,Q^*) = \eps^*$, and for any partition $P,Q$ of $S$ it holds that $\dfd(P,Q)\ge\eps^*$. 
In \Cref{clm:longest-edge} below, we show that $\eps^*$ is the furthest nearest neighbor distance, i.e., the length of the longest edge in the Nearest Neighbor Graph of $S$.

Denote by $\NNG(S)$ the Nearest Neighbor Graph (NNG) of $S$, i.e., the graph whose vertices are the points of $S$, and there is an edge $\{u,v\}$ is the graph if and only if $v$ is a nearest neighbor of $u$ in $S$. Note that a point $u$ can have more than a single nearest neighbor. In this case, we break ties by taking the point with the largest index to be the unique nearest neighbor. It is well-known that when applying such a tie-breaking rule, the NNG is a forest, and a subgraph of the Euclidean minimum spanning tree.

\begin{claim}\label{clm:longest-edge}
    Let $L$ be the longest edge in $\NNG(S)$. Then $\eps^* = L$.
\end{claim}
\begin{proof}
    Notice that $\NNG(S)$ is a spanning forest of $G_L$ that does not contain any vertex of degree $0$ (every point in $S$ has a unique nearest neighbor). Therefore, by \Cref{cor:decision-alg} we get that there exists a partition $P,Q$ of $S$ with $\dfd(P,Q)\le L$.
    
    Assume by contradiction that there is a partition $P',Q'$ of $S$ with $\dfd(P',Q')=\eps^* < L$, and consider a paired-walk $\pi$ along $P$ and $Q$ with cost $\eps^*$. Then for any pair of points $v,u \in S$ that are matched in $\pi$, we have $\|v-u\| < L$.
    Let $\{w, x\}$ be the longest edge in $\NNG(s)$, so $\|w-x\| = L$, and $w,x$ are not matched in $\pi$. Therefore, there exists $w',x'\in S$ such that $w,w'$ are matched in $\pi$ and $x,x'$ are matched in $\pi$. We get that $\|w-w'\| < L$, so $x$ is not a nearest neighbor of $w$, and $\|x-x'\| < L$, so $w$ is not a nearest neighbor of $x$, a contradiction to $\{w, x\}$ being an edge of $\NNG(s)$.
\end{proof}

\vspace{-2pt}
By \Cref{clm:longest-edge}, all the edges of $\NNG(S)$ have length smaller or equal to the optimal distance $\eps^*$, and therefore it is a spanning forest of $G_{\eps^*}$. Thus, given $\NNG(S)$ as an input, the algorithm from \Cref{cor:decision-alg} runs in $O(n)$ time. For $d=2$, computing $\text{NNG}(S)$ can be dome in $O(n \log n)$ time \cite{ep97}. For general dimension $d$, computing $\text{NNG}(S)$ can be done in $O(2^{O(d)} n \log n)$ time \cite{vaidya89}, so the overall running time for our problem is $O(n \log n)$ for any fixed dimension $d$.

We, therefore, obtain the following theorem.
\begin{theorem}\label{thm:discrete-frechet-TSP}
    Given a set $S$ of $n$ points in $\reals^d$, for any fixed dimension $d$, one can find two curves $P,Q$ that partition $S$ such that $\dfd(P,Q)$ is minimized in $O(n\log n)$ time.
\end{theorem}

In the two remarks below, we suggest other ways to use \Cref{clm:longest-edge} and \Cref{cor:decision-alg} for computing an optimal partition of $S$.

\begin{remark}[Using the Net and Prune framework] Har-Peled and Raichel~\cite{HR15} show that the furthest nearest neighbor distance (which, by \Cref{clm:longest-edge}, equals $\eps^*$) can be computed in expected linear time for points in any dimension $d$. Then, to apply \Cref{cor:decision-alg}, we need to compute a spanning forest $T$ of $\Geps$ such that no vertex in $T$ has degree $0$. For this we can use the Net and Prune framework of ~\cite{HR15} as follows. Put $S$ in a grid with cell diameter $\eps^*$. Points that belong to the same grid cell form a clique in $\Geps$, and can be connected, for example, by some star graph. The lonely points, i.e., points that are alone in their cell, can be connected to a point in a neighbor cell (such a point must exist by the way we chose $\eps^*$). Since the number of neighbor cells is $2^{O(d)}$, this results in an $O(2^d\cdot n)$ time algorithm for computing the spanning forest, and $O(2^{O(d)}\cdot n)$ expected time for computing the partition.
\end{remark}

\begin{remark}[Using the MST] 
For points in dimension $d >2$, we can use the the Minimum Spanning Tree of $S$ ($\MST(S)$) instead of $\NNG(S)$ as follows. We compute $\MST(S)$, and then iterate the edges from the longest to shortest. If removing an edge from $\MST(S)$ does not create a vertex of degree $0$, remove it, and otherwise stop. Since $\NNG(S)\subseteq \MST(S)$, and because in $\NNG(S)$ there is no vertex of degree $0$, the resulting graph contains $\NNG(S)$ and the length of its edges is at most $\eps^*$. We can then execute the algorithm from \Cref{cor:decision-alg} on this residue graph. Since $\MST(S)$ can be computed in $O(n^{2-2/{\lceil d/2 \rceil + 1}}+\eps)$ time~\cite{agarwal90}, this is also the total running time of the algorithm. Notice that for $d \le \log {\frac{n^k}{\log n}}$, where $0 < k < 1-1/{O(\log\frac{n}{\log n})}$, we prefer to run the algorithm that finds the residue graph over computing $\MST(S)$. This is because the running time for finding the residue graph is $O(n^{2-2/{\lceil k \cdot \log\frac{n}{\log n}/2 \rceil + 1}}+\eps)$ while computing $\NNG(S)$ is in $O(n^{1+k})$.
\end{remark}

\section{Minimizing the lengths}\label{sec:min-length}

In this section, we focus on problems \ref{prb:min-max-F-TSP} and \ref{prb:min-sum-F-TSP}. Let $S$ be a set of $n > 2$ points in the plane. Our goal is to find a partition $P, Q$ with $\dfd(P, Q) = \eps^*$. In problem \ref{prb:min-max-F-TSP}, the partition also minimizes $\max\{\ell(P),\ell(Q)\}$. In problem \ref{prb:min-sum-F-TSP} it also minimizes $\ell(P)+\ell(Q)$.

Let $\TSP(S)$ be a path on $S$ of minimum length. Denote by $\ell(G)$ the sum of edge lengths of a graph $G$ embedded in the plane.

\vspace{5pt}
\subsection{Min-max discrete \frechet-TSP}
In this section, we prove the following theorem.

\begin{theorem} \label{thm:min-max-length}
    Given a set $S$ of $n$ points in the plane, two curves $P,Q$ that partition $S$ such that $\dfd(P,Q) = \eps^*$ and $\max\{\ell(P), \ell(Q)\}\le 2.75\cdot \ell(\TSP(S))$ can be found in $O(n\log n)$ time.
\end{theorem}

By \Cref{thm:discrete-frechet-TSP}, computing the value $\eps^*$ can be done in $O(n\log n)$ time.
Let $\MST(S)$ be a minimum spanning tree of $S$, and observe that $\ell(\MST(S)) \le \ell(\TSP(S))$. 
We thus focus on computing two curves $P,Q$ that partition $S$ such that $\dfd(P,Q) = \eps^*$ and $\max\{\ell(P), \ell(Q)\} \le 2.75 \cdot \ell(\MST(S))$.
For simplicity, we say that two points (or vertices) $u,v$ are \dfn{close} if $\|u-v\| \le \eps^*$, and otherwise we say that $u,v$ are \dfn{far}.
\begin{observation}\label{obs:close_neighbor}
    Every vertex in $\MST(S)$ has at least one neighbor in $\MST(S)$ that is close to it.
\end{observation}
\begin{proof}
    Assume by contradiction that all neighbors of a vertex $v$ in $\MST(S)$ are far from $v$. By \Cref{lem:g_eps_no_degree_0_exists_division}, $v$ has a neighbor $u$ in $G_{\eps^*}$, and therefore $u$ is close to $v$. This contradicts the fact that the MST contains a nearest neighbor for each vertex (by Kruskal's algorithm).
\end{proof}

\vspace{5pt}

    We now show how to construct two curves $P$ and $Q$ that partition $S$ by traversing $\MST(S)$ in a DFS order, starting from a leaf vertex of $\MST(S)$ as a root, and coloring the vertices \textbf{red} and \textbf{blue}. The blue vertices will be in $P$, and the red vertices in $Q$. The order of the points along the curves would be the same order in which they were colored during the algorithm.

     Consider a vertex $v$ in the rooted tree $\MST(S)$. 
     Let $I_\text{close}(v)$ be the set of children of $v$ that are close to $v$, and similarly, $I_\text{far}(v)$ will be the set of children of $v$ that are far from $v$. In addition, let $L(v)=\{u\mid u \in I_\text{close}(v), I_\text{close}(u)=\emptyset\}$. In other words, $L(v)$ is the set of \dfn{lonely children} of $v$ --- those that are either leaves in $\MST(S)$, or that do not have close children --- and therefore must be in $v$'s star.
     Notice that by \Cref{obs:close_neighbor}, all the nodes in $L(v)$ are close to $v$.
     
     Let $v$ be the current vertex visited by the DFS algorithm.
     The invariant of our recursive DFS algorithm is that if $L(v)$ is empty, then there must be at least one vertex in $I_\text{close}(v)$.
     The algorithm has two main steps: in the first step we color vertices in red and blue, and the second step contains the recursive calls.

    \begin{figure}[h]
        \centering
        \includegraphics[page=1,scale=1.6]{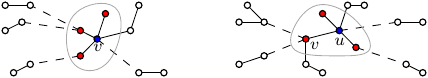}
        \caption{Coloring of $\MST(S)$. Dashed edges correspond to far neighbors. The star to which $v$ belongs is marked in gray.}
        \label{fig:min-length-coloring}
    \end{figure}

    \oldparagraph{Coloring.} In this step, we color the entire star that $v$ belongs to using Algorithm \ref{alg:coloring}. 
    Note that $v$ can be either a center or a leaf in that star. Also, Algorithm \ref{alg:coloring} colors exactly one vertex in \textbf{blue} and the rest in \textbf{red}.
    We progress in recursion on the vertices of the tree. In each recursive call, we activate the coloring algorithm on newly colored vertices in the following order:
     \begin{enumerate}
            \item First, activate the algorithm on children of vertices that are colored \textbf{red}, in reversed order (last colored first called).
            \item then, activate the algorithm on children of the \textbf{blue} vertex.
        \end{enumerate}     
    
    \begin{algorithm}[H]  
    \caption{Coloring for Min-Length}   \label{alg:coloring}
    \KwIn{Vertex $v$}
    \KwOut{Color updates}    
    \eIf{$L(v) \neq \emptyset$}{        
        color the leaf nodes in $L(v)$ \textbf{red}\;        
        color all the other nodes in $L(v)$ \textbf{red}\;        
    }{        
        Let $u$ be an arbitrary vertex in $I_{\text{close}}(v)$\;    
        color $v$ \textbf{red}\;        
        color $u$ \textbf{blue}\;        
        color the leaf nodes in $L(u)$ \textbf{red}\;        
        color all the other nodes in $L(u)$ \textbf{red}\;        
    }
    \end{algorithm}  
    
    \vspace{5pt}
    First, notice that when the algorithm is called with a vertex $v$ in a recursive step, then $v$ and its entire subtree were not colored yet. Moreover, we show that the following invariant holds in each step of the algorithm.
    \begin{claim}
        In each step of the algorithm, if $L(v)$ is empty, then there must be at least one vertex in $I_{close}(v)$.    
    \end{claim}
    \begin{proof}
        Assume by contradiction that in some step of the algorithm both $L(v)$ and $I_\text{close}(v)$ are empty. Then, by \Cref{obs:close_neighbor}, $v$ must have a parent $t$ in $\MST(S)$ such that $v \in I_\text{close}(t)$, and $t$ was already colored by the algorithm.
    Moreover, $v$ is in $L(t)$, because $I_\text{close}(v)=\emptyset$. Therefore, it is not possible that $t$ was colored \textbf{blue}, because then $v$ would have been already colored \textbf{red} and would not be called recursively. 
    In addition, $t$ is not a lonely child of its parent, because $I_\text{close}(t)\neq \emptyset$. Therefore, if $t$ was colored \textbf{red}, then it must be in step 2(a) of the algorithm, which means that $L(t)$ is empty, but this is not possible because $v\in L(t)$.
    \end{proof}

    Let $P=\{p_1,\dots p_k\}$ (resp. $Q=\{q_1,\dots q_m\}$) be the set of vertices that were colored \textbf{blue} (resp. \textbf{red}), in the order in which they were colored during the DFS scan.
    
    \begin{claim}\label{clm:running_time}
        The running time for computing $P$ and $Q$ is $O(n\log n)$.
    \end{claim}
    \begin{proof}
        Calculating the Euclidean MST for the set $S$ takes $O(n\log n)$ time.  
        A standard DFS traversal over the MST also requires $O(n)$ time.
        The additional overhead in the algorithm comes from recursively scanning the $L(v)$ children of each vertex $v$.  
        Since $|L(v)| \le 5$ in the EMST, this adds at most $O(n)$ additional operations. Thus, the total running time remains $O(n\log n)$.
    \end{proof}
    
    \vspace{5pt}
    \begin{claim}\label{clm:small_distance}
        $\dfd(P,Q) \le \eps^*$.
    \end{claim}
    \begin{proof}
        We show that in each step of the algorithm, we color a star in $\MST(S)$ with edges of length at most $\eps^*$. The center is colored \textbf{blue}, and the leaves \textbf{red}. Let $v$ be the current vertex. There are two cases:
        \begin{itemize}
            \item If $L(v)$ is not empty, then $v$ is colored \textbf{blue} and the nodes in $L(v)$ are colored \textbf{red}. Since $L(v) \subseteq I_\text{close}(v)$ (by \Cref{obs:close_neighbor}), this red-blue star has edges of length at most $\eps^*$.
            \item If $L(v)=\emptyset$, then the algorithm picks a vertex $u \in I_\text{close}(v)$ and color it \textbf{blue}. This vertex becomes the center of a star with its children $v\cup L(u)$ which are colored \textbf{red}. Since $L(v) \subseteq I_\text{close}(u)$, we again obtain a red-blue star with edges of length at most $\eps^*$.        
        \end{itemize}
        Since the vertices of $P$ and $Q$ are ordered by the time they were colored, we get that the sequence of starts corresponds to a paired-walk of cost at most $\eps^*$ between $P$ and $Q$, as required.
    \end{proof}

    Next, we bound the lengths of the curves $P$ and $Q$ in relation to $\MST(S)$. 
    Notice that the order in which we color the \textbf{blue} vertices (the vertices of $P$) follows a classic DFS preordering, and therefore we clearly have $\ell(P)\le 2\cdot \ell(\MST(S))$. However, the order in which we color the \textbf{red} vertices slightly differs from a classic DFS preordering. However, since all the children of $v$ that are colored in this step are close to $v$, the additional traversal overhead is small. The following claim together with \Cref{clm:small_distance} implies \Cref{thm:min-max-length}.

    \begin{figure}[!ht]
        \centering
        \includegraphics[page=4,scale=1.6]{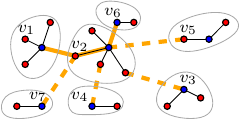}
        \caption{An illustration of the algorithm. The root vertex is $v_1$, and the algorithm runs on $v_1,\dots,v_7$ in this order. The stars are marked in gray, and the dashed edges mark far vertices. The orange edges are the set $W$ of edges that connect the stars.}
        \label{fig:DFS_alg}
    \end{figure}

    \begin{claim}\label{clm:approx}
        $\max\{\ell(P),\ell(Q)\}\le 2.75 \cdot \ell(\MST(S))$.
    \end{claim}
    \begin{proof}
        Let $T$ be the path on all the points in $S$, which is obtained by traversing the edges of $\MST(S)$ in the order in which the algorithm colors the vertices, regardless of their color. Clearly $\ell(T)\ge \max\{\ell(P),\ell(Q)\}$. We show that $\ell(T)\le 2.75\cdot\ell(\MST(S))$.

        Recall that the degree of any vertex in $\MST(S)$ is at most $5$. Since we choose the root of $\MST(S)$ to be a leaf, then for every vertex $v$ in the rooted tree, $L(v)$ may contain at most $4$ vertices.
        Consider a step of the algorithm were the current vertex is $v$. 
        If $v$ was colored \textbf{blue}, denote $L(v) = u_1, \dots, u_k$ for $1\le k\le 4$. Then, the subpath of $T$ that was added in this step is $T_v = \{v,u_1,v,u_2,v,\dots,u_k\}$. This is because, in $T$, we are moving from $u_i$ to $u_{i+1}$ through $v$, since $\{u_i,u_{i+1}\}$ is not an edge in $\MST(S)$. 
        If $v$ was colored \textbf{red}, the algorithm picked a vertex $u \in I_{close}(v)$. Denote $L(u) = u_1, \dots, u_k$ for $1\le k\le 4$. The subpath of $T$ that was added in this step is $T_v=\{v,u,u_1,u,u_2,u,\dots,u_k\}$.
        
        We show how to charge the edges of $\MST(S)$ for each such star-subpath and for each of the edges that are connecting between star-subpaths.
        First, notice that in each star-subpath $T_v$, each edge of the star is traversed at most twice.
        Let $U = \cup_{x \in P} T_x$ be the set of edges that are charged for the star-subpaths themselves. Because the edges of the stars are disjoint, each edge in $U$ is charged at most twice.
        
        Let $W$ be the set of edges that are traversed in $T$ when connecting between any two star-subpaths. Each edge of $W$ is charged at most twice because we traverse the stars following the classic DFS order. Therefore, in both sets $U, W$ each edge is charged at most twice. Notice that $U, W$ are distinct.
        
        The edges in $U$ have length at most $\eps^*$, so edges of length larger than $\eps^*$ can only appear in $W$. In addition, edges that are incident to leaves in $\MST(S)$ appear only in $U$, because they do not connect between star-subpaths. We conclude that leaf-edges are charged only twice, and edges of length larger than $\eps^*$ are also charged only twice.

        Consider an edge $\{v,u\} \in T_v$ such that $L(v)$ is empty and $u \in I_{close}(v)$. Notice that the edge $\{v, u\}$ is charged only once in $U$ (it appears once in $T_v$) and once in $W$ (to connect $u$ with the centers of stars in $v$'s subtree), so in total it is charged at most twice.   
        
        The only case left to handle is edges $\{x,u_i\}$ in a star such that $u_i \in L(x)$ ($x$ can be either the current vertex or its child $u\in I_\text{close}(v)$). In this case, $I_{close}(u_i)=\emptyset$. If $u_i$ is not a leaf, then there exists an edge $\{u_i, w_i\}$ such that $w_i\in I_{far}(u_i)$. Since the algorithm recursively runs on $w_i$, the edges $\{u_i,w_i\}$ is in $W$. In other words, for every $\{x, u_i\}\in U$ there exists an edge $\{u_i, w_i\}\in W \setminus U$ such that $\|u_i-w_i\|> \eps^*$, which is charged only twice. Therefore, if $\{x,u_i\}$ is charged four times (twice in $U$ and twice in $W$), we charge it three times, and transfer the forth charge to $\{u_i, w_i\}$, so that both edges are charged only three times in total. Note that no other edge can transfer its charge to $\{u_i, w_i\}$, because $w_i$ has to come after $u_i$ in the order of the traversal.
        
        This gives a bound of $3 \cdot \ell (\MST(S))$. To further improve the bound, notice that the last edge $\{x, u_k\}$ is charged only once in $U$.
        Moreover, it is also charged only once in $W$, because the children in its subtree are the first to be called recursively in this step, so $\{x, u_k\}$ only appears on the subpath that goes back to $x$.
        In the worst case, when $k = 4$, we charge the first three edges $\{x,u_1\}, \{x,u_2\}, \{x,u_3\}$ a total of $3$ times (by transferring one charge to the corresponding edge in $W$), and $\{x, u_4\}$ is charged only twice.
        By choosing the farthest child to be colored last, i.e., $\|x-u_4\|\ge\max\{\|x-u_1\|, \|x-u_2\|, \|x-u_3\|\}$, each of the first three edges can transfer a charge of $\tfrac{1}{4}$ to $\{x, u_4\}$, resulting in a total charge of at most $2 \tfrac{3}{4}$ per edge of $\MST(S)$, as claimed.
    \end{proof}

\noindent \Cref{thm:min-max-length} follows from \Cref{clm:running_time}, \Cref{clm:small_distance}, and \Cref{clm:approx}.

\subsection{Min-sum discrete \frechet-TSP}
In this section we consider the case where the goal is to minimize the sum of the lengths. The following theorem is a corollary of \Cref{thm:min-max-length}.

\begin{theorem}\label{thm:min-sum-length}
    Given a set $S$ of $n$ points in the plane, two curves $P,Q$ that partition $S$ such that $\dfd(P,Q) = \eps^*$ and $\ell(P) + \ell(Q)\le 4.75\cdot \ell(\MST(S))$ can be found in $O(n \log n)$ time.
\end{theorem}

\begin{proof}
    Let $P$ and $Q$ be the curves obtained by the algorithm in the previous section. Then by \Cref{thm:min-max-length} $\ell(Q) \le 2.75 \cdot \ell(\MST(S))$. 
    For $P$, note that it follows a DFS traversal of the $\MST(S)$ restricted to the centers. 
    Since each edge is used at most twice and no vertex is repeated as in $Q$, this gives $\ell(P) \le 2 \cdot \ell(\MST(S))$. 
    Combining the two bounds yields $\ell(P) + \ell(Q) \le 4.75 \cdot \ell(\MST(S))$.
\end{proof}

\noindent This result can be improved by using a $(1+\eps)$ approximation for TSP (for $\eps>0$) on one of the curves. We make use of the algorithm of Kisfaludi-Bak, Nederlof and Węgrzycki \cite{KNW25} which computes in $2^{O(\varepsilon^{1-d})} n \log n$ time.
We get the following result.

\begin{theorem}
    Given a set $S$ of $n$ points in $\reals^d$, two curves $P,Q$ that partition $S$ such that $\ell(P) + \ell(Q)\le (4+\eps) \cdot \ell(\MST(S))$ can be found in $O(n \log n)$ time for constant $d, \varepsilon$.
\end{theorem}
\begin{proof}
    Given the set of stars $S_1,\dots,S_k$ that were obtained in the proof of \Cref{lem:g_eps_no_degree_0_exists_division}, each star $S_i$ has a \textbf{blue} vertex that is connected in $S_i$ to a set of \textbf{red} vertices. 
    We run an algorithm that computes a $(1+\eps)$ approximation for the TSP. On the \textbf{blue} vertices, and obtain a curve $P$ with the length $(1+\eps)\cdot\ell(\TSP(S))$.
    
    We then construct the curve $Q$ on the \textbf{red} vertices following the order of the \textbf{blue} vertices in $P$, i.e., for each vertex $v\in P$ that corresponds to the star $S_i$, we connect the red vertices of $S_i$, and then connect these red paths according to the TSP order on the \textbf{blue} vertices. We can bound the length of $Q$ by the length of a curve traversing $P$ with detours for traversing the \textbf{red} vertices. The sum of lengths of these detours is bounded by traversing each edge of a star twice, which is bounded by $2\cdot\ell(\MST(S))$. We therefore get that $\ell(Q) \le \ell(P) + 2\cdot\ell(\TSP(S))$ and thus $\ell(Q) + \ell(P) \le (4+\eps)\cdot\ell(\TSP(S))$.
    \end{proof}

\subsection{Comparing to the optimal solution}
In previous sections, we compared the solution obtained from our algorithm to $\TSP(S)$. We now show that it gives a constant approximation comparing to the optimal solution. 

\begin{observation}\label{obs:tsp-s-less-sum-any-dividion}
    Let $P,Q$ be two curves that partition $S$, such that $\dfd(P, Q) = \eps$ for some $\eps > 0$. Then $\ell(\TSP(S))\le \ell(P)+\ell(Q)+\eps$.
\end{observation}
\begin{proof}
    By concatenating $P$ and $Q$, we obtain a path of at most $\ell(P) + \ell(Q)+\eps$ that traverses all the points of $S$, which is a feasible solution for TSP on $S$.
\end{proof}

\noindent Denote by $P^*, Q^*$ two curves that partition $S$ with $\dfd(P^*,Q^*) = \eps^*$ such that $\max \{\ell(P^*),\ell(Q^*)\}$ is minimized.

\begin{lemma} \label{lem:bound_eps_by_curve_lengths}
    Let $P,Q$ be two curves that partition $S$, such that $\dfd(P,Q) = \dfd(P^*,Q^*) = \eps^*$.
    Denote by $L_P$ (reps. $L_Q$) the maximum length of an edge in $P$ (resp. $Q$). If $|P|,|Q|\ge 2$, then $\eps^*\le L=\max\{L_P,L_Q\}$.
\end{lemma}
\begin{proof}
    Assume by contradiction that $\eps^*> L$. That means that all edges of both $P$ and $Q$ are of strictly smaller length than $\eps^*$.
    Denote $P=\{p_1,\dots,p_k\}$ and $Q=\{q_1,\dots,q_m\}$.
    If $k\ge 2$, then the discrete \frechet\ distance between $P_\text{odd}=\{p_1,p_3,p_5,\dots\}$ and $P_\text{even}=\{p_2,p_4,p_6,\dots\}$ is at most $L$. Similarly, if $m\ge 2$, then the \frechet\ distance between $Q_\text{odd}=\{q_1,q_3,q_5,\dots\}$ and $Q_\text{even}=\{q_2,q_4,q_6,\dots\}$ is at most $L$. Therefore, the discrete \frechet\ distance between $P_\text{odd}\circ Q_\text{odd}$ and $P_\text{even}\circ Q_\text{even}$ is at most $L$, in contraction to the optimality of $\eps^*$.
\end{proof}

Note that if one of $P, Q$ is a single vertex, the above lemma may not be correct. Let $S = \{p,q_1,q_2,\dots\}$ such that $\|p-q_i\| = \eps$ for some $\eps > 0$. Let $\|q_i-q_j\|\ll\eps$ for all $q_i, q_j\in Q$. Then $\eps$ is the optimal \frechet\ distance. And yet, the lemma does not hold for $P=\{p\}$ and $Q=\{q_1,q_2,\dots\}$.

An immediate corollary of \Cref{obs:tsp-s-less-sum-any-dividion} and \Cref{lem:bound_eps_by_curve_lengths} is that $\ell(\TSP(S))\le \ell(P)+\ell(Q)+\eps\le 3\cdot \max\{\ell(P),\ell(Q)\}$, and $\ell(\TSP(S))\le 2\cdot (\ell(P)+\ell(Q))$.
Therefore, by \Cref{thm:min-max-length} we have the follwiing corrolary.
\begin{corollary}
    \label{clr:bound-max-curve-length}
    Given a set $S$ of $n$ points in the plane, we can find in $O(n\log n)$ time two curves $P$, $Q$ that partition $S$ such that $\dfd(P,Q) = \eps^*$ and 
    $\ell(P)+\ell(Q)=O(\ell(P^*)+\ell(Q^*))$, or $\max\{\ell(P), \ell(Q)\}=O(\max\{\ell(P^*),\ell(Q^*)\})$.
\end{corollary}

\section{Balancing the number of vertices}\label{sec:balance}

In this section, we address \Cref{prb:balanced-F-TSP} (balanced-\frechet-TSP) under the discrete \frechet\ distance.
Clearly, it is possible that there is no pair of curves that partition $S$ and have both $\dfd(P,Q)\le \eps^*$ (the optimal distance) and $\max\{|P|,|Q|\}=\lceil n/2\rceil$ (see, e.g. \Cref{fig:no_balance}). 

We show that for points in the plane our solution to discrete-\frechet-TSP can be adjusted such that $\max\{|P|,|Q|\}\le \lfloor n/2\rfloor+2$. In other words, we prove the following theorem.

\begin{theorem}\label{thm:balanced}
    Given a set $S$ of points in the plane, there always exists two curves $P,Q$ that partition $S$, such that $\dfd(P,Q)\le \eps^*$, $|P|\ge |Q|$ and $|P|-|Q|\le 4$ (and in case that $n$ is odd, $|P|-|Q|\le 3$). Moreover, such curves can be found in $O(n\log n)$ time.
\end{theorem}
\begin{proof}
    Consider the set of stars $S_1,\dots,S_k$ that were obtained in the proof of \Cref{lem:g_eps_no_degree_0_exists_division}, and that can be computed in $O(n\log n)$ time by \Cref{thm:discrete-frechet-TSP}). Each star $S_i$ has a blue vertex $\textbf{blue}(S_i)$ which is connected in the spanning tree $T_C$ to a set of red vertices $\textbf{red}(S_i)$. For \Cref{thm:discrete-frechet-TSP} we are using the nearest neighbor graph $\NNG(S)$, applying the unique nearest neighbor rule, and therefore the maximum degree in $\NNG(S)$ is at most $5$. This is due to kissing number in Euclidean space.

    Therefore, for every $1\le i\le k$, we have $\textbf{red}(S_i)\le 5$.
    
    Notice that flipping the colors of the vertices in a star $S_i$ does not change the correctness of the algorithm, and we still obtain two curves that partition $S$ and have distance at most $\eps^*$. Therefore, we can perform the following procedure. Let $P$ and $Q$ be the curves obtained from the algorithm of \Cref{thm:discrete-frechet-TSP}, then by construction $|Q|\ge|P|$. 
    Set $w=|Q|-|P|$, and iterate over the stars $S_1,\dots,S_k$: if $w\ge 5$, flip the colors in the current star, and update $w$. The algorithm terminates when $w\le 4$, or after $S_k$ is flipped.

    Since initially each star $S_i$ has exactly one vertex in $P$ (its center $\textbf{blue}(S_i)$) and at most $5$ vertices in $Q$ ($\textbf{red}(S_i)$), a flip adds $c=|\textbf{red}(S_i)|-1\le 4$ vertices to $P$, and removes $c$ vertices from $Q$, which in total reduces $w$ by $2c\le 8$. Therefore, after a flip we have $w\ge -3$.
    Assume by contradiction that the algorithm terminates with $w\ge 5$. Then all the stars $S_1,\dots,S_k$ were flipped, but this is a contradiction because in the first step we had $w=|P|-|Q|>0$, and thus after flipping all the stars we have $w=|Q|-|P|<0$.

    We conclude that when the algorithm terminates, we have $-3\le w\le 4$. If $w<0$, we flip all the stars, and get $0\le w= |Q|-|P|\le 4$, as required. Finally, note if $n$ is odd, then $w$ must be odd, and therefore in this case we have $0\le w= |Q|-|P|\le 3$.
\end{proof}

\begin{remark}
    In fact, the maximum degree of the NNG for points in $d$ dimensions is equal the kissing number of spheres in $d$ dimensions, and therefore \Cref{thm:balanced} can be generalized to higher dimensions accordingly.
\end{remark}

\subsection{Relaxing the distance requirement}
The example in \Cref{fig:no_balance} shows a set $S$ for which there is no balanced partition with distance $\eps^*$. However, notice that if we relax the requirement on the distance between the curves and allow it to be up to $2\eps^*$, then we can split each star into smaller stars by connecting pairs of leaves, so that the maximum degree of a star becomes two, which allows for a balanced partition. Below we show that by allowing an even smaller relaxation (the discrete \frechet\ distance will be at most $\sqrt{3}\cdot\eps^*$), we can always obtain optimally balanced curves.

\begin{figure}[h!]
    \centering
    \includegraphics[page=1,scale=0.7]{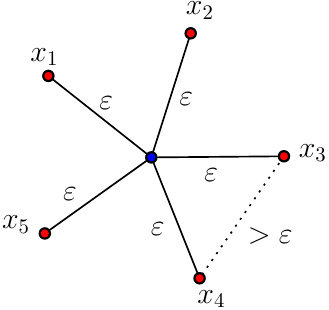}
    \caption{A set of $6$ points for which the optimal solution for \frechet-TSP is $\eps$, $|P|=1$ and $|Q|=5$.}
    \label{fig:no_balance}
\end{figure}

\begin{theorem}\label{thm:balanced-relxed-distance}
    Given a set $S$ of points in the plane, there always exists two curves $P, Q$ that partition $S$, such that $\dfd(P,Q) \le \sqrt3\cdot\eps^*$ and $\max\{|P|,|Q|\}=\lceil n/2 \rceil$. Moreover, such curves can be found in $O(n\log n)$ time.
\end{theorem}
\begin{proof}
    Consider the set of stars $S_1,\dots,S_k$ that were obtained in the proof of \Cref{lem:g_eps_no_degree_0_exists_division}. If the degree of each star is at most two (i.e. $|\textbf{red}(S_i)|\le 2$ for every $1\le i\le k$), then by applying arguments similar to the proof of \Cref{thm:balanced}, we get that $\max\{|P|,|Q|\}=\lceil n/2 \rceil$ as required.
    
    Otherwise, let $S_i$ be a star such that $|\textbf{red}(S_i)|\ge 3$. 
    Let $c=\textbf{blue}(S_i)$, then there is at least one pair of nodes $u,v\in \textbf{red}(S_i)$ such that the smaller angle $\angle ucv$ at at most $\frac{2\pi}{3}$.
    Thus, by the law of cosines, the distance between them is $$\|v-u\| = \sqrt{\|c-v\|^2+\|c-u\|^2- 2\cdot \cos(\frac{2\pi}{3}) \cdot \|c-v\|\cdot \|c-u\|}.$$ Since the length of any edge in $S_i$ is at most $\eps^*$, we have $$\|v-u\| \leq \sqrt{2(\eps^*)^2 - 2(\eps^*)^2\cos(\frac{2\pi}{3})} = \sqrt{2(\eps^*)^2 + (\eps^*)^2} = \sqrt{3}\cdot\eps^*.$$
    We now split $S_i$ into two stars: remove $v,u$ from $\textbf{red}(S_i)$ and create a new star, $S_{k+1}$, with $\textbf{red}(S_{k+1})= \{v\}$ and $\textbf{blue}(S_{k+1})=u$.
        
    We continue this process until all our stars have degree at most two.
    Notice that any star that we add consists of a single pair of points with distance at most $\sqrt{3}\cdot\eps^*$. Therefore, when applying the arguments from the proof of \Cref{thm:balanced} as before, we obtain $\max\{|P|,|Q|\}=\lceil n/2 \rceil$, and $\dfd(P,Q) \le \sqrt3\cdot\eps^*$.
\end{proof}

\begin{figure}[h!]
    \centering
    \includegraphics[page=2,scale=0.7]{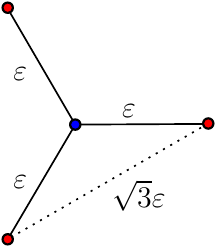}
    \caption{A set of $4$ points for which the optimal solution for \frechet-TSP is $\eps$, $|P|=1$ and $|Q|=3$. By allowing the distance to be $\sqrt3\eps$, we can obtain a solution where $|P|=|Q|=3$.}
    \label{fig:best_relax}
\end{figure}
The example in \Cref{fig:best_relax} shows that the bounds in \Cref{thm:balanced} and \Cref{thm:balanced-relxed-distance} are tight, as we conclude in the observation below.
\begin{observation}
    There exists a set $S$ of points in the plane such that: (i) for any two curves $P,Q$ that partition $S$ and have $\dfd(P,Q)\le\eps^*$ it holds that $|P|-|Q|\ge 4$, and (ii) for any two curves $P,Q$ that partition $S$ and have $|P|-|Q|= 0$ it holds that $\dfd(P,Q)\le\sqrt3\cdot\eps^*$.
\end{observation}

\section{Tighter balancing}
In the proof of \Cref{thm:balanced} we flip the colors of some stars in order to obtain a difference of four between $|P|$ and $|Q|$. However, for some instances there might still be a different way to arrange the stars (i.e., a different paired-walk) that allows for a more balanced partition. Therefore, to further improve the difference, we now show how to split some of the starts (as we do in the proof of \Cref{thm:balanced-relxed-distance}) in order to reduce the difference between $|P|$ and $|Q|$, while still having $\dfd(P,Q)\le\eps^*$.

Given a star subgraph $H$ of $\Gepso$, define its \dfn{balance} as $\b(H)=\textbf{deg}(H)-1$, where $\textbf{deg}(H)$ is the degree of the center node of $H$. 
We call a star $H$ with balance $\b(H)$ a $\b(H)$-star (e.g., a 0-star, a 1-star, etc.). For a set $X$ of star subgraphs of $\Gepso$, denote $\b(X)=\sum_{H\in X}\b(H)$.

Let $\alpha=\{H_1,\dots,H_t\}$ be a set of disjoint stars $H_1,\dots,H_t$ in $\Gepso$ that together cover $S$. We call such a set a \dfn{star cover} of $S$. Consider a valid coloring of the stars in $\alpha$ in two colors, then as in the proof of \Cref{lem:g_eps_no_degree_0_exists_division} this coloring yields two curves that partition $S$ and a paired-walk along them with cost at most $\eps$. Intuitively, $\b(S_i)$ is the amount that a star contributes to either $|P|$ and $|Q|$. Therefore, the following problem is equivalent to \Cref{prb:balanced-F-TSP} (balanced-\frechet-TSP).
\begin{problem}\label{prb:star-cover}
    Given a set $S$ of $n$ points in the plane and a value $\eps>0$, find a star cover $\alpha=H_1,\dots,H_t$ of $S$, that can be partitioned into two sets of stars $\rho$ and $\beta$, that minimizes the difference $$\Delta_\alpha(\rho,\beta)=\left|\b(\rho)-\b(\beta)\right|.$$
\end{problem}

We thus focus on solving the problem defined above. For a given star cover $\alpha$, an \dfn{optimally balanced partition} of $\alpha$ is a partition of $\alpha$ into two sets of stars $\rho$ and $\beta$ that minimizes $\Delta_\alpha\coloneq\Delta_\alpha(\rho,\beta)$.
\fullversion{In \Cref{sec:tighter-balancing}}

Below, we show that when $|S|=n$ is even,
any star cover $\alpha$ of $S$ with $\Delta_\alpha=4$ consists of only $0$-stars and $4$-stars.\fullversion{(see \Cref{{lem:D=4 cant have smaller stars}}).} We then present an algorithm that runs in $O(n^2)$ time, and either returns a star cover $\alpha$ of $S$ with $\Delta_{\alpha} \le 2$, or reports (correctly) that the optimal star cover has $\Delta_{\alpha}=4$.
\begin{restatable}{theorem}{balancedEven}  \label{thm:balanced-even}
    Given a set $S$ of $n$ points in the plane, where $n$ is even, there exists an algorithm that runs in $O(n^2)$ time, which returns a star cover $\alpha$ of $S$ s.t. either:
    \begin{enumerate}
        \item $\Delta_\alpha = 4$ and it is the minimum possible value of any star cover of $S$ (i.e. $\alpha$ must be optimal).
        \item $\Delta_\alpha \leq 2$.
    \end{enumerate}
    
\end{restatable}

\newenvironment{caseenum}{%
  \begin{enumerate}[{\bfseries {Case} 1 -}]
  \setlength{\leftskip}{-1.5cm}
}{%
  \end{enumerate}
}

Observe that if the maximum degree of the stars in a star cover $\alpha$ is constant, then an optimally balanced partition of $\alpha$ can be found in linear time, because in this case it reduces to a bounded variant of the subset sum problem \cite{P99linear}.
Therefore, by taking $\alpha$ to be the set of stars from the proof of \Cref{lem:g_eps_no_degree_0_exists_division} (in which the maximum degree is $5$), we can compute in linear time an optimally balanced partition of $\alpha$.
By \Cref{thm:balanced}, we have $\Delta_\alpha\le 4$ (there always exists a coloring with a difference of $4$).
Note that for even values of $n$, we have $\Delta_\alpha \in \{0,2,4\}$, and for odd values of $n$ we have $\Delta_\alpha \in \{1,3\}$.
We thus conclude that if $\Delta_\alpha \in \{0,1\}$, then $\alpha$ already gives an optimal solution for \Cref{prb:star-cover}. We therefore assume that $2\le \Delta_\alpha\le 4$.

Denote by $W$ the set of edges of stars in $\alpha$. Clearly, if $W$ consist of all the edges in $\Gepso$, then $\alpha$ is the only possible star cover for $S$, and by the optimality of the partition of $\alpha$, it is also an optimal solution for \Cref{prb:star-cover}.
We therefore assume that there is at least one edge in $\Gepso$ which is not in $W$.

Since $\alpha$ is partitioned optimally and $\Delta_\alpha>1$, the only way to obtain a better solution for \Cref{prb:star-cover} is by using a different set of stars.
In other words, the optimal solution for the problem must contain edges from $\Gepso\setminus E(\alpha)$.
We give algorithms to find this optimal solution (if exists). For convenience, we name the kissing number in Euclidean space as \dfn{kissing property}.

The following auxiliary lemma will be useful in two of our main theorems below.

\begin{lemma}\label{lem:balance:0-star-permutation-reduction}
        Let $\alpha_0$ be a star cover of a set $\hat{S}\subseteq S$, that consists only of $0$-stars. Then given $\Gepso$ as an input, there is an algorithm that runs in $O(|\Gepso|)$ time and either finds a star cover $\hat{\alpha}$ of $\hat{S}$ and a partition $\rho,\beta$ of $\hat{\alpha}$ with $\Delta_\alpha(\rho,\beta)=2$, or reports (correctly) that any star cover $\hat{\alpha}$ of $\hat{S}$ consists of only $0$-stars (and thus for any partition $\rho,\beta$ of  $\hat{\alpha}$ we have $\Delta_\alpha(\rho,\beta)=0$).
\end{lemma}
\begin{proof}
        First notice that $|\hat{S}|$ must be even. We show that there exists a star cover $\hat{\alpha}$ of $\hat{S}$ and a partition $\rho,\beta$ of $\hat{\alpha}$ with $\Delta_\alpha(\rho,\beta)=2$ if and only if one of the following conditions hold:
        \begin{enumerate}
            \item There exists two $0$-stars $S_i=\{u_1,v_1\},S_j=\{u_2,v_2\}$ in $\alpha_0$ such that $\|u_1-u_2\|\le \eps$ and $\|u_1-v_2\|\le \eps$.
            \item There exists three $0$-stars $S_i=\{u_1,v_1\},S_j=\{u_2,v_2\},S_k=\{u_3,v_3\}$ in $\alpha_0$ such that $\|u_1-u_2\|\le \eps$ and $\|v_1-v_3\|\le \eps$.
        \end{enumerate}
\begin{figure}[h!]
    \centering
    \includegraphics[scale=1]{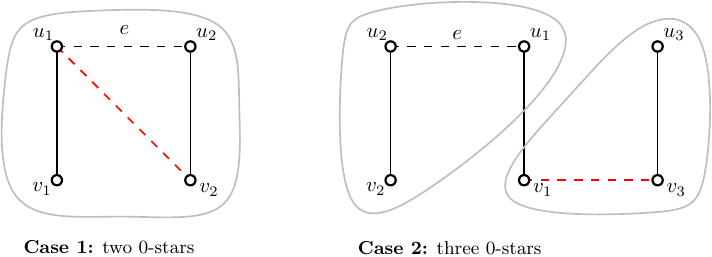}
    \caption{$a^0$ star partition cases}
    \label{fig:o-star-permutation-cases}
\end{figure}

First we show that if one of the conditions hold then there exists a star cover $\hat{\alpha}$ of $\hat{S}$ and a partition $\rho,\beta$ of $\hat{\alpha}$ with $\Delta_\alpha(\rho,\beta)=2$.
In the first case (\Cref{fig:o-star-permutation-cases}, left) we remove the edge $\{u_2,v_2\}$ and connect the vertices $u_2,v_2$ to $S_i$, so $S_i$ becomes a 2-star. This result in a star cover with a single 2-star and some number of 0-stars, which can be partitioned into $\rho,\beta$ with $\Delta_\alpha(\rho,\beta)=2$.
For the second case, we remove the edge $\{u_1,v_1\}$ and connect $u_1$ to $u_2$ and $v_1$ to $v_3$, which result in two 1-stars instead. Therefore the new set of stars can be partitioned into $\rho,\beta$ with $\Delta_\alpha(\rho,\beta)=2$.

For the second direction, assume that there exists a star partition $\hat{\alpha}$ of $\hat{S}$ and a partition $\rho,\beta$ of $\hat{\alpha}$ with $\Delta_\alpha(\rho,\beta)=2$. Then there exists at least one edge $e=\{u_1,u_2\}\notin E(\hat{\alpha})$ which belongs to a non-0-star. Let $S_i = \{u_1, v_1\}$ and $S_j = \{u_2, v_2\}$ be the stars in $\hat{\alpha}$ that cover the vertices $u_1,u_2$. Assume by contradiction that both conditions do not hold, so $\min\{\|u_1-v_2\|,\|u_2-v_1\|\}>\eps^*$, there is no vertex $v_3\neq v_2$ such that $\|v_1 - v_3 \| \le \eps^*$ and no vertex $v_4\neq v_1$ such that $\|v_2-v_4\|\}\le\eps^*$.

Notice that $E(\hat{\alpha})$ can contain at most one of the edges $\{u_1, v_1\},\{u_2, v_2\}$. If $E(\hat{\alpha})$ contains only the edge $\{u_2, v_2\}$, then $v_1$ must be in another star, but there is no other vertex at distance at most $\eps^*$ from it. Similarly, it cannot be that $E(\hat{\alpha})$ contains only the edge $\{u_1, v_1\}$.
If both edges $\{u_2,v_2\},\{u_1,v_1\}$ are not in $E(\hat{\alpha})$, then it must be that $\{v_1,v_2\}$ is a 0-star in $\hat{\alpha}$, and because $e$ is not a 0-star, there must be some other edge connecting either $u_1$ of $u_2$ to some other vertex $v\notin\{v_1,v_2,u_1,u_2\}$ which belongs to some star $S_k$ in $\alpha^0$. In this case the second condition holds for $S_i$, $S_j$ and $S_k$.
\end{proof}

The observation below will be useful when combining two partial covers of $S$.
\begin{observation}\label{obs:balance-sum-delta-of-subsets-deltas}
    Let $S$ be a set of points, and $S_1,S_2$ a partition of $S$ into two sets of points. Let $\alpha_1,\alpha_2$ be star covers of $S_1,S_2$, respectively, with partitions $\rho_1,\beta_1$ of $\alpha_1$ and $\rho_2,\beta_2$ of $\alpha_2$. Let $\alpha=\alpha_1\cup\alpha_2$, then
    \begin{itemize}
        \item $\rho=\rho_1\cup \rho_2, \beta=\beta_1\cup \beta_2$ is a partition of $\alpha$ with $\Delta_\alpha(\rho,\beta)=\Delta_\alpha(\rho_1,\beta_1)+\Delta_\alpha(\rho_2,\beta_2)$, and 
        \item $\rho=\rho_1\cup \beta_2, \beta=\beta_1\cup \rho_2$ is a partition of $\alpha$ with $\Delta_\alpha(\rho,\beta)=\Delta_\alpha(\rho_1,\beta_1)-\Delta_\alpha(\rho_2,\beta_2)$.
    \end{itemize}
\end{observation}

\noindent We now turn to inspect the different cases for the value of $\Delta_\alpha$.

\subsection{Even Number of Points}
In this section, we handle the case in which $\Delta_\alpha = 4$. Notice that this means that $n$ is even.
Below, we present an algorithm that runs in $O(n^2)$ time which either decides that $\alpha$ is an optimal solution for \Cref{prb:star-cover}, or finds a star partition $\alpha'$ for which $\Delta_{\alpha'} < 4$. We begin by showing that there can only be two types of stars in such star cover $\alpha$.

\begin{lemma}\label{lem:D=4 cant have smaller stars}
    Let $\alpha$ be a star cover of $S$, with $\Delta_\alpha=4$. Then there is no star $S_i \in \alpha$ such that $\b(S_i)\in \{1,2,3\}$. In other words, $\alpha$ only consists of $0$-stars and $4$-stars.
\end{lemma}
\begin{proof}
    Let $\rho, \beta$ be an optimally balanced partition of $\alpha$. Assume w.l.o.g. that $\b(\rho) > \b(\beta)$. 
    Let $c \in \{1,2,3\}$. Assume by contradiction that $\rho$ contains a $c$-star $S_i$.
    By removing $S_i$ from $\rho$ and adding it to $\beta$, we obtain another partition of $\alpha$ into two sets $\rho' = \rho\setminus\{S_i\}$ and $\beta'= \beta\cup\{S_i\}$, where $\Delta_\alpha(\rho',\beta')=|\b(\rho')-\b(\beta')|=|\b(\rho)-c-(\b(\beta)+c)|=|\b(\rho)-\b(\beta)-2c|=|4-2c| < 4$, a contradiction the optimality of the partition $\rho, \beta$. We thus conclude that $\rho$ can contain only $0$-stars and $4$-stars.
    
    Since $\b(\rho) > \b(\beta)$, $\rho$ must contain at least one $4$-star, $S_k$.
    Assume by contradiction that $\beta$ contains a $c$-star $S_j$, and consider the partition of $\alpha$ into two sets $\rho''=(\rho\setminus\{S_k\})\cup\{S_j\}$ and $\beta''=(\beta\setminus\{S_j\})\cup\{S_k\}$, where $\Delta_\alpha(\rho'',\beta'')=|\b(\rho'')-\b(\beta'')|=|\b(\rho)-4+c-(\b(\beta)-c+4)|=|\b(\rho)-\b(\beta)+2c|=|4-2c| < 4$. Again, this is a contradiction to the optimality of the partition $\rho, \beta$.
\end{proof}

\balancedEven*\label{thm:balanced-even-proof}
\begin{proof}
    We start by computing some star cover $\alpha$ for $S$, with an optimal partition $\rho, \beta$. If $\Delta_\alpha < 4$ then $\Delta_\alpha \leq 2$ and we are done.
    Let $\Delta_\alpha = 4$. We run the following case analysis on the edges of $\Gepso\setminus E(\alpha)$ (see \Cref{fig:balanced-cases-even}). In each case, we show that there is only a certain local change to the stars in $\alpha$ that would result in another star cover $\alpha'$ with $\Delta_{\alpha'} \le 2$. That is, if such specific local change is not possible then $\alpha$ is optimal.
\begin{figure}
    \centering
    \includegraphics[scale=0.7]{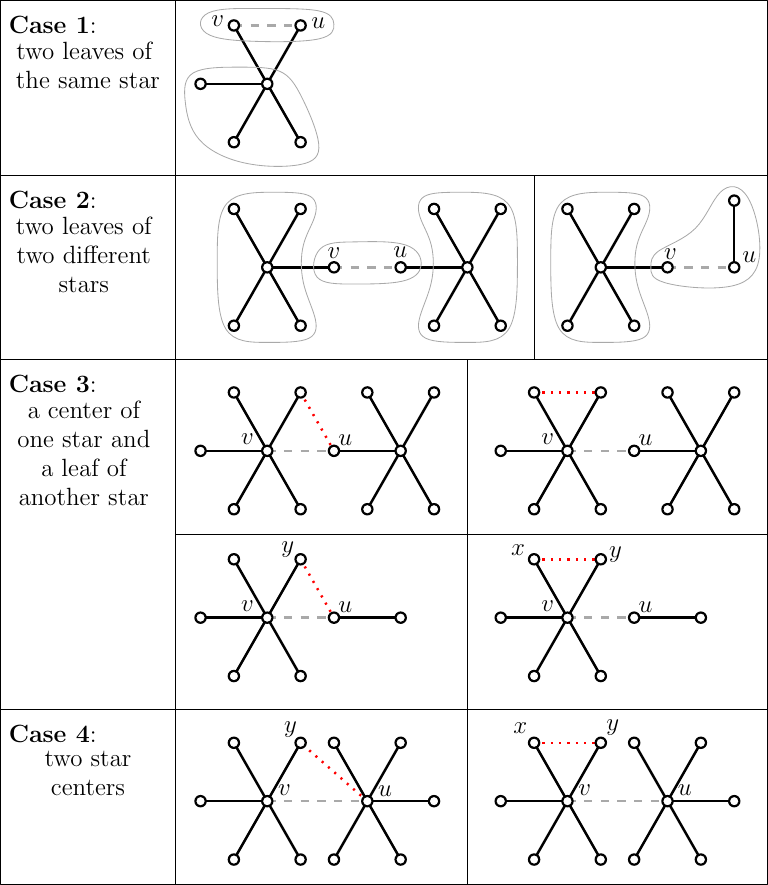}
    \caption{Even number of points - Balance cases on 4-stars}
    \label{fig:balanced-cases-even}
\end{figure}

    There are four different cases: 
    \begin{enumerate}[{\bfseries {Case} 1 -}]
    \setlength{\itemsep}{0.5em}
        \item \label{4d:case:leaf-same-star} There exists an edge $\{u,v\}\in \Gepso\setminus E(\alpha)$ which connects two leaves of a $4$-star $S_i$. In this case we remove $u,v$ from $S_i$, and add $\{u,v\}$ as a new $0$-star to $\alpha$. After removing $u,v$ from $S_i$, it becomes a $2$-star. 
        \item \label{4d:case:leaf-diff-star} There exists an edge $\{u,v\}\in \Gepso\setminus E(\alpha)$ which connects a leaf $v$ of a $4$-star $S_i$ with a leaf $u$ 
        of a star $S_j$. We remove $v$ from $S_i$, so it becomes a $3$-star.
        \begin{enumerate}[(i)]
            \item If $S_j$ is a $4$-star, then we also remove $u$ from $S_j$ (so it becomes a $3$-star), and we add $\{u,v\}$ as a new $0$-star to $\alpha$.
            \item Else, $S_j$ is a $0$-star, so we add $u$ to $S_j$ and it becomes a $1$-star.
        \end{enumerate}
        \item \label{4d:case:leaf-center} There exists an edge $\{u,v\} \in \Gepso\setminus E(\alpha)$ which connects a center $v$ of a $4$-star $S_i$ with a leaf $u$ of a star $S_j$. In this case, we get that $S_i\cup\{v,u\}$ is a $5$-star, which is a star of degree $6$. By the \dfn{kissing property}, there must be an edge $\{x,y\}$ in $\Gepso\setminus E(\alpha)$ between two leaves of that star. 
        \begin{enumerate}[(i)]
            \item If one of $x,y$ is $u$, then since $u$ is a leaf of $S_j$ it means that we are in case \ref{4d:case:leaf-diff-star}.
            \item Else, if both $x,y$ are leaves of $S_i$, the we are in Case \ref{4d:case:leaf-same-star}.
        \end{enumerate}
        \item \label{4d:case:center-center} There exists an edge $\{u,v\} \in \Gepso \setminus E(\alpha)$ which connects a center $v$ of a $4$-star $S_i$ with a center $u$ of a $4$-star $S_j$. Similar to the previous case, we get that $S_i\cup\{v,u\}$ is a star of degree $6$, and thus there must be an edge $\{x,y\}$ between two leaves of that star.
        \begin{enumerate}[(i)]
            \item If one of $x,y$ is $u$, then since $u$ is a center of $S_j$ it means that we are in case \ref{4d:case:leaf-center}.
            \item Else, if both $x,y$ are leaves of $S_i$, the we are in case \ref{4d:case:leaf-same-star}.
        \end{enumerate} 
        \item \textbf{0-Stars - } If there is an edge $\{u,v\}\in \Gepso \setminus E(\alpha)$ that belongs to one of these cases, then the updated star cover $\alpha'$ contains a $c$-star for some $c\in\{1,2,3\}$, and thus by \Cref{lem:D=4 cant have smaller stars} it cannot have $\Delta_{\alpha'}=4$. Therefore, since $n$ is even, we get $\Delta_{\alpha'}\le 2$. Otherwise, if no edge of $\Gepso \setminus E(\alpha)$ belongs to one of the above case, then the only edges in $\Gepso \setminus E(\alpha)$ are edges that connect vertices of $0$-stars.    
        Denote by $\alpha_0\subseteq\alpha$ the set of $0$-stars from $\alpha$. Then $\alpha_4=\alpha\setminus\alpha_0$ must have $\Delta_{\alpha^4}=4$, meaning that $\alpha_4$ consists of an odd number of $4$-stars. Moreover, note that any star cover of $S$ must contain $\alpha_4$. 
        Let $\hat{S}$ be the set of vertices of stars from $\alpha_0$. We apply \Cref{lem:balance:0-star-permutation-reduction} on $\alpha_0$ and $\hat{S}$. If we got a star cover $\hat{\alpha}$ of $\hat{S}$ with a partition $\rho,\beta$ such that $\Delta_{\hat{\alpha}}(\rho,\beta)=2$, then by \Cref{obs:balance-sum-delta-of-subsets-deltas}, $\alpha'=\alpha_4 \cup \hat{\alpha}$ is a star cover of $S$ with $\Delta_{\alpha'}\le 2$. Otherwise, by \Cref{lem:balance:0-star-permutation-reduction}, any star cover $\hat{\alpha}$ over $\hat{S}$ has $\Delta_{\hat{\alpha}}=0$, and since any star cover of $S$ must contain $\alpha^4$, there is no star cover $\alpha$ over $S$ with $\Delta_{\alpha} < 4$.        
    \end{enumerate}    
    To check if one of the cases above occurs, we simply scan all the edges of $\Gepso \setminus E(\alpha)$ which can be done in $O(|\Gepso|)$ time.
    The running time of the algorithm is $O(n^2)$, as each modification is local.
\end{proof}

\section{Continuous \frechet-TSP is hard}\label{sec:continuous}
As mentioned in the introduction, Buchin and Kilgus~\cite{BK22} show that the continuous \frechet\ distance between 2 point sets is NP-hard.
Similarly, we show that the \frechet-TSP problem under the continuous \frechet\ distance is also NP-hard. The reduction is quite similar to the reduction described in~\cite{BK22}, and for completeness, we include a sketch of the proof.
We reduce from a restricted version of 3-SAT, called $(3,B2)$-SAT, in which each clause contains exactly three literals, and each literal (that is, a variable or its negation) appears exactly twice in the entire formula. As in~\cite{BK22}, we assume that in the $(3,B2)$-SAT instance no two clauses share more than one literal (this is possible by the construction in \cite{BKS03}).

Given a $(3,B2)$-SAT formula, we construct a set $S$ of points in the plane such that the formula is satisfiable if and only if there exist two curves $P,Q$ that partition $S$ and have $\df(P,Q)\le \eps$.

Our construction is as follows. Each clause in the formula is represented by a single point, where three line segments intersect, with one segment for each literal in the clause.
This is possible because no two clauses share more than one literal. Since each literal appears at most twice in the formula, the line segment corresponding to a literal passes through the two points representing the clauses containing that literal.
The variable gadgets are then created in a way for each clause point, one curve must visit it, and the other curves must pass close to it. This is possible because no two clauses share more than one literal. Since each literal appears in at most two clauses, the line segment corresponding to a literal passes through the two points representing the clauses containing that literal. 
This ensures that each clause is ``covered'' by the appropriate combination of curves according to the truth assignment encoded by the variable gadgets.

To make this construction concrete, we now describe the variable gadgets in detail. For every input variable $x$ in the formula, we build a gadget similar to the one described in \cite{BK22}. 
Each gadget has two points $x_1, x_2$ that represent the positive literal, and two points $\bar x_1, \bar x_2$ that represent the negative literal, see \Cref{fig:Variable X gadget} below.

\begin{figure}[ht]
    \centering
    \includegraphics[page=1,,scale=0.8]{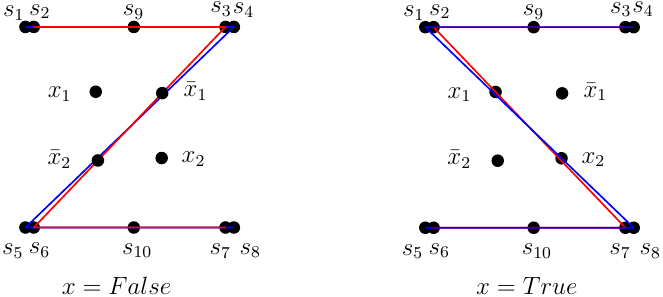}
    \caption{Variable $x$ gadget}
    \label{fig:Variable X gadget}
\end{figure}

The key difference for the 1-set case, compared to the original two–point-set case in ~\cite{BK22}, is that the vertices are not yet assigned to the curves.
In order to force the partition to be as in~\cite{BK22}, we place the points at the corners of the variable gadgets at distance $\eps>0$, for small enough $\eps$. If the formula is stisfiable, then the optimal \frechet\ distance will be $\eps$,
because each pair of corner points can be matched while still passing through all the clause points, as illustrated in \Cref{fig:Variable X gadget}.

More precisely, the gadgets consists of four pairs of points $(s_1,s_2)$, $(s_3,s_4)$, $(s_5,s_6)$ and $(s_7,s_8)$, each pair of points is placed so the distance between them is very small $\eps$. This ensures that both curves are forced to pass through all four corners, and, in order to preserve the uniqueness of each point along the curves, they cannot complete a full loop around the gadget and return to their starting corner. As a result, the traversal of the gadget is restricted, so that each curve may follow only one of the two diagonals. This restriction allows the gadget to enforce that exactly one of the line segments corresponding to $x$ or $\bar{x}$ is visited by one of the curves (and the other one passes close enough to it), corresponding to the chosen assignment of the 3SAT variable.

For the clauses in the $(3,B2)$-Linear-SAT formula, we intersect the 3 different gadgets in one of the $x_1, x_2, \bar x_1, \bar x_2$ representing the variables in the clause, this is constructed in the same way as \cite{BK22}, as once the variable gadget are built the difference of the points being preassigned a curve or not does not affect the construction anymore, \Cref{fig:sat-to-1-set-continuous-frechet}.

\begin{figure}[h]
    \centering
    \includegraphics[page=2,,scale=0.8]{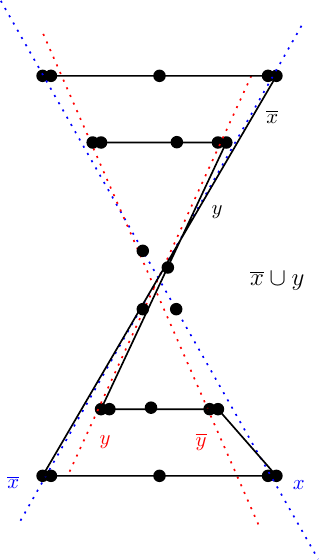}
    \caption{$(3,B2)$-SAT to 1 set continuous \frechet\ problem}
    \label{fig:sat-to-1-set-continuous-frechet}
\end{figure}

\begin{theorem}
    There is no polynomial-time solution for uniquely partitioning $S$ into curves $P,Q$ with minimal \frechet\ distance, unless $P=NP$.
\end{theorem}
\begin{remark}
    Partitioning $S$ into curves $P,Q$ with minimal weak \frechet\ distance, also cannot be done in polynomial time unless $P=NP$. This follows directly from the proof over the \frechet\ distance that also applies for the week case.
    Notice that the proof does not apply for the non-unique case, where we allow the same point to repeat on the same curve, but not in both because then we can just take identical curves over all set $S$.
\end{remark}

\newpage

\bibliography{refs}

\appendix

\section{Visiting a point more than once}\label{sec:multiset}
Consider the following variant of the problem, in which we allow a point to be used more than once by the same agent.

\begin{problem}\label{prb:non-unique-F-TSP}
    Given a set $S$ of points in $\reals^d$, find two curves $P$ and $Q$ such that each point in $S$ is a vertex of either $P$ or $Q$ (but not both,and it can appear as a vertex multiple times in one curve), such that $\delta(P,Q)$ is minimized.
\end{problem}

With this variant we may obtain closer routes under the continuous \frechet\ distance as depicted in \Cref{fig:continuous-non-unique}. 
\begin{figure}[ht]
    \centering
    \includegraphics[width=0.3\linewidth]{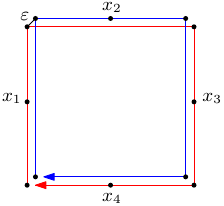}
    \caption{The optimal (continuous) \frechet\ distance that can be achieved when an agent can use a point multiple times is $\eps$. For the distance to be $\eps$, each of the $x_i$'s needs to be matched to an edge of the other curve whose endpoints are above and below it, or to the right and left of it. This is not possible if every point is visited exactly once.}
    \label{fig:continuous-non-unique}
\end{figure}
However, for the discrete distance the two version are equivalent, as we show in \Cref{obs:multiset} below.
\begin{observation}\label{obs:multiset}
       There exist two curves $P,Q$ such that partition $S$ with repetitions (i.e., a point in $S$ can appear more than once in $P$ or $Q$, but it cannot be in both curves) and have $\dfd(P,Q) \le \eps$ if and only if there exist two curves $P',Q'$ that uniquely partition $S$ and have $\dfd(P',Q') \le \eps$. 
\end{observation}
\begin{proof}
    One direction is trivial. For the other direction, let $P,Q$ be two curves that partition $S$ non-uniquely and have $\dfd(P,Q) \le \eps$, and assume by contradiction that there is no unique partition $P',Q'$ of $S$ such that $\dfd(P',Q') \le \eps$. 
    By \Cref{lem:g_eps_no_degree_0_exists_division}, there exists a vertex $v$ of degree $0$ in $G_\eps$. Assume w.l.o.g. that $v \in P$. Since $v \notin Q$ and $\|v- u\| > \eps$ for every $u\in S$, we get that $\dfd(P,Q) > \eps$, a contradiction.
\end{proof}

Nevertheless, when the goal is to minimize the lengths of the curves, then these versions are not equivalent, as illustrate in \Cref{fig:optimal-frechet-unique-nonunique-curve-length-example} below.
\begin{figure}[h]
    \centering
    \includegraphics[width=0.3\linewidth]{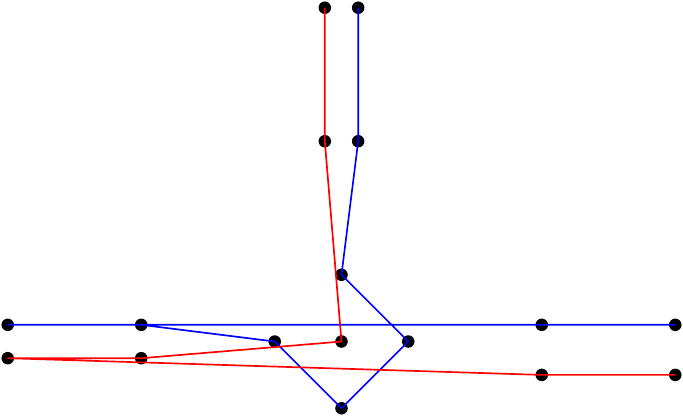} 
    \hspace{1cm}
    \includegraphics[width=0.3\linewidth]{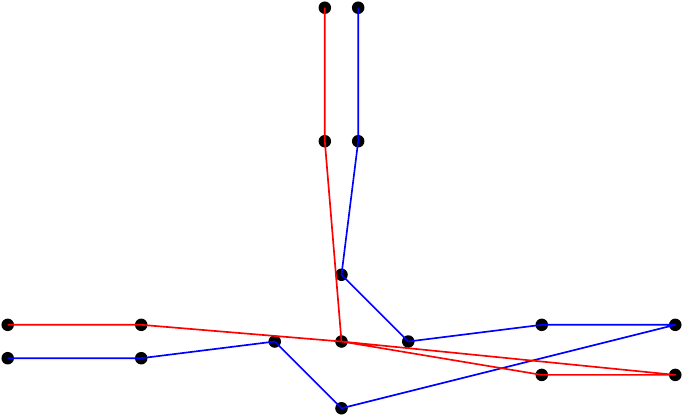}
    \caption{Left: the optimal discrete \frechet\ distance when no repetitions allowed. Right: the length is smaller when repetitions are allowed.}
    \label{fig:optimal-frechet-unique-nonunique-curve-length-example}
\end{figure}

\end{document}